\documentclass[letterpaper,twocolumn,10pt]{article}

\usepackage{usenix, epsfig, endnotes}

\usepackage{amsmath}
\usepackage{amssymb}
\usepackage{enumerate}
\usepackage{graphicx}
\usepackage{subfigure}
\usepackage{amsthm}
\usepackage{hyperref}
\usepackage[ruled,vlined]{algorithm2e}
\usepackage{bm}
\usepackage{verbatim}

\usepackage[noend]{algpseudocode}
\usepackage{xspace}
\usepackage{url}
\usepackage{breakurl}
\usepackage{color}
\usepackage{xcolor}
\usepackage{soul}
\usepackage{listings}
\usepackage{multirow}
\usepackage{tabularx}

\usepackage{flushend}

\newcommand{\sect}[1]{Section~\ref{#1}\xspace}

\newcommand{\sysname}{PrivPy\xspace}
\newcommand{\para}[1]{\smallskip\noindent\textbf{#1}}

\newcommand{\tabincell}[2]{\begin{tabular}{@{}#1@{}}#2\end{tabular}}
\newcommand{\code}[1]{\texttt{#1}\xspace}
\newcommand{\cl}[1]{\code}

\newcommand{\server}[1]{$\mathbf{S_#1}$}

\newcommand{\share}[1]{[\![#1]\!]}
\newcommand{\map}[1]{\widetilde{#1}}

\newcommand{\alert}[1]{{\color{black} #1}}

\makeatletter
\newcounter{protocol}
\newenvironment{protocol}[1][htb]
  {
   \let\c@algocf\c@protocol
   \begin{algorithm}[#1]%
       \small
  }{\end{algorithm}}
\makeatother

\SetKwInput{KwInput}{Input}
\SetKwInput{KwOutput}{Output}
\SetKwInput{KwStep}{Steps}

\lstdefinestyle{mystyle}{
	keywordstyle=\color{magenta},
	basicstyle=\ttfamily\footnotesize,
	breakatwhitespace=false,         
	breaklines=true,                 
	captionpos=b,                    
	keepspaces=true,                 
	numbers=left,                    
	numbersep=5pt,                  
	showspaces=false,                
	showstringspaces=false,
	showtabs=flase,                  
	tabsize=2,
}
\newtheorem{theo}{Theorem}

\newtheorem{defi}{Definition}

\begin{document}
\title{\sysname: Enabling Scalable and General Privacy-Preserving Machine Learning}

\author{Yi Li$^1$, Yitao Duan$^2$, Yu Yu$^3$, Shuoyao Zhao$^3$ and Wei Xu$^1$ \\
{\em  $^1$Tsinghua University \quad
             $^2$Netease Youdao \quad
             $^3$Shanghai Jiaotong University} \\ [2mm]
}
\date{}

\maketitle

\begin{abstract}
	Privacy is a big hurdle for collaborative data mining across multiple parties.  We present multi-party computation (MPC) framework designed for large-scale data mining tasks.  \sysname combines an easy-to-use and highly flexible Python programming interface with state-of-the-art secret-sharing-based MPC backend.  With essential data types and operations (such as NumPy arrays and broadcasting), as well as automatic code-rewriting, programmers can write modern data mining algorithms conveniently in familiar Python.  We demonstrate that we can support many real-world machine learning algorithms (e.g. logistic regression and convolutional neural networks) and large datasets (e.g. 5000-by-1-million matrix) with minimal algorithm porting effort. 
\end{abstract}

\section{Introduction}
\label{section:intro}

Privacy is an important issue in big data age. 
The success of data mining is often built on data, and it is often desirable to 
integrate data from multiple sources for better mining results. 
However, the unrestricted exchanging of sensitive data may threaten users' privacy 
and is often prohibited by laws or business practices. 
How to protect privacy while allowing the integration of multiple data sources demands prompt solutions. 

Secure multi-party computation (MPC) allows players to collectively compute a function without 
revealing private information except for the final output.  
MPC often uses various cryptographic primitives, such as garbled circuit~\cite{YAO_PROTOCOL} and secret sharing~\cite{SHAMIR}, 
with different efficiency and security assumptions.
After more than 30 years of development, we have started to see that the real-world data mining applications start to use MPC~\cite{PEM, PPDM_STATUS}.
However, numerous challenges still exist that prevent widespread adoption of secure computation
techniques. 


One of the most important issue hindering MPC's adoption is programmability, especially for ``big data'' applications.  
Despite of the development of efficiency improvement of MPC during the past decades, existing MPC solutions often ignore the core requirements of data mining applications. They either require considerable expertise in cryptography to understand the cost of each operation, or use special programming languages with high learning curves~\cite{L1, TASTY, SPDZ, SECUREML, SECUREC, OBLIVM}.  
Some useful solutions, such as \cite{EMP}, though providing rich interfaces for MPC, mainly focus on basic MPC operations, including not only basic arithmetics but also low-level cryptography tools such as \emph{oblivious transfer}~\cite{OT}. 
In contrast, machine learning programmers use Python-based frameworks like PyTorch~\cite{PYTORCH}, Tensorflow~\cite{TENSORFLOW} and Scikit-learn~\cite{SCIKIT_LEARN} with built-in support of high-level data types like real numbers, vectors and matrices, as well as  non-linear functions such as the logistic function and ReLu.  
It is almost impossible for data scientists to rebuild and optimize all these often taken-for-granted primitives in a modern machine learning package in an MPC language.  On the other hand, it is also costly for MPC experts to rewrite all the machine learning algorithm packages.  Thus, it is essential to design an MPC front-end that is friendly with the data mining community, which is Python with NumPy~\cite{NUMPY} nowadays.  
Actually, many machine learning frameworks use Python front-ends and provide Numpy-style array operations to ease machine learning programming.

In this paper, we propose \sysname, an efficient framework for privacy-preserving collaborative 
data mining, aiming to provide an elegant end-to-end solution for data mining programming. 
The \sysname front-end provides Python interfaces that resemble those from NumPy, one of the most popular Python packages, as well as a wide range of functions commonly used in machine learning. 
We also provide an computation engine which is based on secret sharing and provides efficient arithmetics.
We would like to stress that the main goal of \sysname is not to make theoretic breakthrough in cryptographic protocols, but rather to build a practical 
system that enables elegant machine learning programming on secure computation frameworks and makes right trade-offs between efficiency and security.
In particular, we make the following contributions:

\begin{enumerate}
\item \textbf{Python programming interface with high-level data types.}  We provide a very clean Python language integration with privacy-enabled common operations and high-level primitives, including \emph{broadcasting} that manipulates arrays of different shapes, and the \emph{ndarray} methods, two Numpy~\cite{NUMPY} features widely utilized to implement machine learning algorithms, with which developers can port complex machine learning algorithms onto \sysname with minimal effort. 

\item \textbf{Automatic code check and optimization.} Our front-end will help the programmers avoid ``performance pitfalls'', by checking the code and optimizing it automatically. 

\item  \textbf{Decoupling programming front-end with computation back-ends.} We introduce a general private operator layer to allow the same interface to support multiple computation back-ends, allowing trade-offs among different performance and security assumptions.  Our current implementation supports the SPDZ back-end, the ABY3 back-end and our own computation engine.  


\item \textbf{Validation on large-scale machine learning tasks.} We demonstrate the practicality of our system for data mining applications, such as data query, logistic regression and convolutional neural network (CNN), on real-world datasets and application scenarios. 
We also show that our system can scale to large-scale tasks by transparently manipulating a 5000-by-1-million matrix.
\end{enumerate}

\section{Related Work}
\label{section:related_work}

In this paper, we mainly focus on privacy-preserving computation systems for general arithmetics, especially for data mining tasks.
A practical such system includes two parts: an efficient computation engine and an easy-to-use programming front-end.

Frameworks based on (fully) homomorphic encryption~\cite{FHE_STATISTICS, HELIB} are impractical due to heavy computation overhead. 
Approaches based on garbled circuit (GC)~\cite{FairplayMP, EMP, OBLIV_C, GC_3PC} will be impractical for general-purpose arithmetical computations, especially for various kinds of machine learning algorithms, as they are costly in bandwidth.
There are also many MPC frameworks using secret sharing supporting general arithmetics.
For example, \cite{IMPROVED_3PC} performs integer/bit multiplication with 1 round and optimal communication cost using three semi-honest servers. 
SPDZ~\cite{SPDZ} uses addition secret sharing and can tolerate up to $n-1$ corrupted parties.
While natively supporting efficient integer operations, most of them (e.g. \cite{SHAREMIND, SPDZ}) support real numbers by parsing each shared integer into $m$ field elements ($m$ is the bit length of the each field element) and use bit-level operations to simulate fixed/floating point operations~\cite{FIXED_POINT, HYBRID, SECURE_FLOATING}, thus requires each party to send $O(m)$ meessages. 
SecureML~\cite{SECUREML} is based on two-party secret sharing and provides built-in fixed-point multiplication with $O(1)$ message complexity, but requires expensive precomputation to generate Beaver multiplication triples. 
ObliviousNN~\cite{OBLIVIOUS_NN} optimizes the performance of dot product, but suffers similar problem with SecureML.
ABY3~\cite{ABY3}, which extends the work of \cite{IMPROVED_3PC} and provides three-party computation, is the state-of-the-art for general arithmetics. 
To perform fixed-point multiplication, ABY3 provides two alternatives: one requires a lightweight precomputation and each party needs to send no more than 2 messages in 1 round in the online phase, while the other requires no precomputation and each party sends no more than 2 messages, but needs 2 rounds.
In comparison, our computation engine, which also provides built-in support for fixed-points, performs fixed-point multiplication in 1 round without precomputation and each party only sends 2 messages.

We emphasize that the adoption of privacy-preserving computation is beyond the computation efficiency and the programmability is as the same importance.
TASTY~\cite{TASTY} and ABY~\cite{ABY} provide interfaces for programmers to convert between different schemes.
However, they only expose low-level interfaces and the programmers should decide by themselves which cryptographic tools to choose and when to convert them, making the learning curve steep.
L1~\cite{L1} is an intermediate language for MPC and supports basic operations. But L1 is a domain-specific language and does not provide high-level primitives to ease array/matrix operations frequently used in machine learning algorithms.
\cite{AUTOMATED_SYNTHESIS} and \cite{ENCRYPTED_CLASSIFICATION} suffer from similar problems.
PICCO~\cite{PICCO} supports additive secret sharing and provides customized C-like interfaces. But the interfaces are not intuitive enough and only support simple operations for array. Also, according to their report, the performance is not practical enough for large-scale arithmetical tasks. 
KSS~\cite{KSS} and ObliVM~\cite{OBLIVM} also suffer from these issues.
\cite{SPDZ_COMPILER} provides a compiler for SPDZ and \cite{GENERALIZING_SPDZ} extends it to support more MPC protocols. But they are still domain-specific and do not provide enough high-level primitives for machine learning tasks.
\sysname, on the other hand, stays compatible with Python and provides high-level primitives (e.g. broadcasting) with automatic code check and optimization, requiring no learning curve on the application programmer side, making it possible to implement machine learning algorithms conveniently in a privacy-preserving situation.

\section{\sysname Design Overview}
\label{section:system_design}
\subsection{Problem formulation}
\label{section:preliminaries}

\para{Application scenarios.  } We identify the following two major application scenarios for privacy-preserving data mining:

\begin{itemize}
    \item \para{multi-source data mining. } It is common that multiple organizations (e.g. hospitals), each independently collecting part of a dataset (e.g. patients' information), want to jointly train a model (e.g. for inferring a disease), without revealing any information. 
    \item \para{inference with secret model and data. } Sometimes the parameters of a model are valuable.  For example, the credit scoring parameters are often kept secret.  Neither the model owner nor the data owner want to leak their data in the computation. 
\end{itemize}

\para{Problem formulation.  } 
We formulate both scenarios as an MPC problem:  there are $n$ clients $C_i (i = 1, 2, \dots, n)$. 
Each $C_i$ has a set of private data $D_i$ as its input.
The goal is to use the union of all $D_i$'s to compute some function $o = f(D_1, D_2, \dots, D_n)$, while no private information other than the output $o$ is revealed during the computation. 
$D_i$ can be records collected independently by $C_i$, and $C_i$'s can use them to jointly train a model or perform data queries.

\para{Security assumptions.  }
\label{section:assumption}
Our design is based on two widely adopted assumptions in the security community~\cite{SHAREMIND, IMPROVED_3PC, PEM}: 
1) All of the servers are \emph{semi-honest}, which means all servers follow the protocol and would not conspire with other servers, but they are curious about the users' privacy and would steal information as much as possible; 
and 2) all communication channels are secure and adversaries cannot see/modify anything in these channels.  
In practice, as there is a growing number of independent and competing cloud providers, it is feasible to find a small number of such servers. 
We leave extensions of detecting malicious adversaries as future work.


\subsection{Design overview}
\label{subsec:architecture}

Fig.~\ref{figure:architecture} shows an overview of \sysname design, which has two main components: the \emph{language front-end} and the \emph{computation engine back-end}.
The front-end provides programming interfaces and code optimizations. The back-end performs the secret-sharing-based privacy-preserving computation.  We discuss our key design rationals in this section.  

\begin{figure}[tb]
	\centering
	\includegraphics[width = 0.45\textwidth]{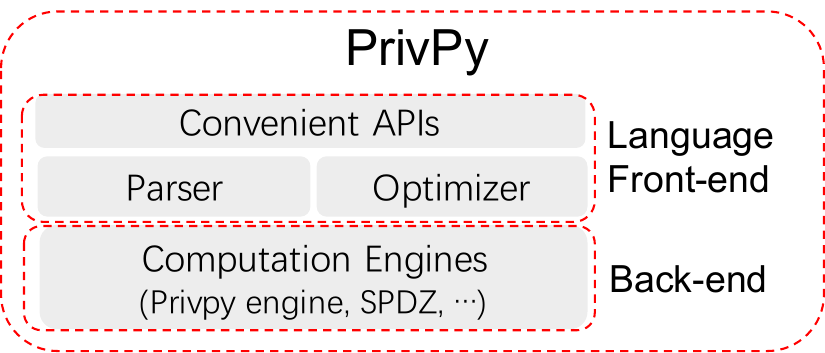}
	\caption{The overview of \sysname architecture.}\label{figure:architecture}
    \vspace{-0.15in}
\end{figure}

\para{Decoupling the frontend with backend.  }
We decouple the front-end and back-end using an extensible interface. The major benefit is that we can adapt to multiple language and backend techniques.  While we believe Python is a natural choice for the frontend for the data mining community, we support multiple MPC backends to allow tradeoffs among different security assumptions and performance.  Our interface between frontend and backend is extensible.  The basic interface only require scaler data types and operations, making it possible to use very simple engines.  Then we add extension interfaces to fully utilize backends with performance-optimizations such as array types and complex computation functions (e.g. vector outer product).  It is an analogous to the extensible instruction set architecture (ISA) design.  In fact, we currently support three backends: our own backend, SPDZ~\cite{SPDZ}, ABY3~\cite{ABY3}.  

\para{Focus on performance optimizations for the entire algorithm.  }
Performance is the key to enable scalable data mining tasks.  We optimize performance at three different levels: 1) optimize single operation performance using the 2-out-of-4 secret sharing protocols; 2) batch up operations whenever possible; 3) perform language-level optimizations in the language frontend.  

\para{Based on 2-out-of-4 secret sharing protocol. }
Handling collaborative data mining tasks over multiple parties has some two unique performance challenge: 1) the computation might happen on wide area networks (WANs), and thus it is both bandwidth and latency sensitive; and 2) there are vast amount of data, and thus we need to minimize the overall computation, including the preprocessing.  Existing engines either require multiple rounds of communication and thus perform poorly in WANs (\cite{SHAREMIND, SPDZ}), or require significant amount of pre-computation (\cite{SECUREML, ABY3}).  

We design a 2-out-of-4 secret sharing protocol combining the ideas in SecureML~\cite{SECUREML} and ABY3~\cite{ABY3}.  By adding a fourth server, we can eliminate the pre-computation in ABY3, but keep its one-round only online communication feature for fixed-point multiplication while preserving the same online communication complexity.  Also, the correctness proof of our protocol directly follows ~\cite{ABY3} and ~\cite{SECUREML}, making it simple to establish the correctness. 

\para{Hierarchical private operations (POs).  }
We call operations on private variable \emph{private operations (POs)}.  Table~\ref{table:operations} provides an overview of different POs implemented in \sysname.  

\begin{table}[tb]
	\centering
    \footnotesize
	\begin{tabular}{|c|c|c|c|c|}
        \hline
        type & \multicolumn{4}{c|}{operations} \\
        \hline
        basic & add/sub & multiplication & \tabincell{c}{oblivious\\ transfer } & \tabincell{c}{bit \\ extraction} \\
		\hline
        \multirow{2}{*}{\tabincell{c}{derived}} & comparison & sigmoid & relu & division \\
        \cline{2-5}
        & log & exp & sqrt & abs \\
        \hline
        \multirow{11}{*}{\tabincell{c}{ndarray}} & all & any & append & argmax \\
        \cline{2-5}
        & argmin & argparition & argsort & clip \\
        \cline{2-5}
        & compress & copy & cumprod & cumsum \\
        \cline{2-5}
        & diag & dot & fill & flatten \\
        \cline{2-5}
        & item & itemset & max & mean \\
        \cline{2-5}
        & min & ones & outer & partition \\
        \cline{2-5}
        & prod & ptp & put & ravel \\
        \cline{2-5}
        & repeat & reshape & resize & searchsorted \\
        \cline{2-5}
        & sort & squeeze & std & sum \\
        \cline{2-5}
        & swapaxes & take & tile & trace \\
        \cline{2-5}
        & transpose & var & zeros &  \\
        \hline
	\end{tabular}
	\caption{Supported operations of the \sysname front-end.}
	\label{table:operations}
\end{table}

We identified a number of POs that are either essential for computation or performance critical, and we implement them directly in secret sharing.   We call them basic POs. Limited by space, we only introduce the fixed-point number multiplication PO in Section~\ref{section:multiplication}.  
Another set of POs we implement is to support vector and matrix operations.  It is essential to support the array types in Python. 

One good feature of the 2-out-of-4 secret sharing is that the result of the computation is still secret shares with exactly the same format.  Thus, we can concatenate different POs together and implement derived POs.  Note that even derived POs performs better than a Python-library as it is pre-compiled in the engine - much like the built-in routines in a database system.

\para{Using frontend to provide both flexibility and performance.  }
PrivPy frontend not only makes it easy for data scientists to write in familiar Python, but also it provide extensive optimizations to support array types, including arbitrary sized arrays and operations.  Also, it automatically performs code rewriting to help programmers to avoid common performance pit-falls.  

\para{Engine architecture.  }
We use four (semi-honest) servers to implement the 2-out-of-4 secret sharing protocol above in the \sysname engine: \server{1}, \server{2}, \server{a} and \server{b}.
We adopt a client/server model, just like many existing MPC systems~\cite{IMPROVED_3PC, PEM, P4P, SHAREMIND, OBLIVM}.
Clients send secretly shared data to the servers, then the servers perform privacy-preserving computation on these shares (see Fig.~\ref{figure:engine}).

Each of the four servers has two subsystems.
The \emph{secret sharing storage (SS store) subsystem} provides (temporary) storage of shares of private inputs and intermediate results. 
while the \emph{private operation (PO) subsystem} provides an execution environment for private operations.  
The servers read shares from the SS store, execute a PO, and write the shares of the result back to the SS store. 
Thus we can compose multiple POs to form a larger PO or a complex algorithm.

\begin{figure}
	\centering
	\includegraphics[width = 0.35\textwidth]{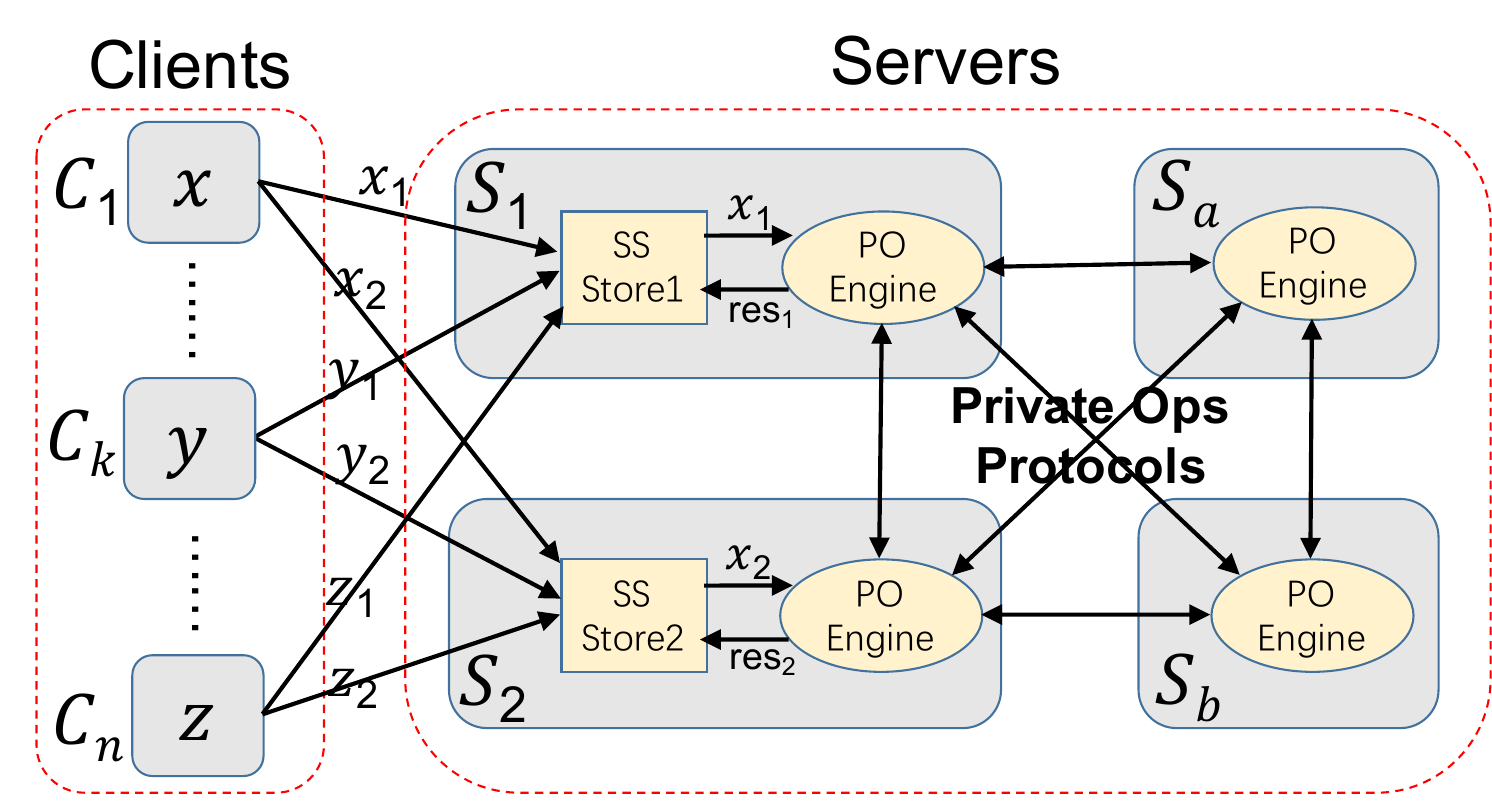}
	\caption{The overview of \sysname computation engine.}\label{figure:engine}
    \vspace{-0.15in}
\end{figure}

\para{Task execution.  }
In summary, PrivPy runs a machine learning tasks in the following four steps:
1) the Python front-end analyzes and rewrites the program for algorithm-level performance optimization.
2) Each client computes the secret shares for her private variables, and sends the resulting shares to the servers.  
3) All servers runs the Python code in parallel on the private shares without any client involvement, until it reached the \code{reveal()} point in the code. 
4) The servers invoke \code{reveal()}, the clients are notified to find the \code{result} shares, and finally recover the cleartext result.

\section{The \sysname computation engine}
\label{section:engine_overview}
In this section, using multiplication as an example, we introduce our secret sharing protocol design, and then we describe the private operation (POs) we support.
The thread model (security assumptions) is defined in \sect{section:preliminaries}.

\subsection{Replicated 2-out-of-4 secret sharing}
\label{section:secret_sharing}

\para{Secret Sharing.  }
Secret sharing encodes a secret number into multiple shares and distributes the shares to a group of participants such that no information about the raw number is revealed as long as no sufficient is gathered. The simplest secret sharing is to encode a number $x$ into two numbers $r$ and $x - r$ where $r$ is a random number. Thus one can reconstruct $x$ only if he gets both shares.

\para{Replicated 2-out-of-4 secret sharing.  }
We combine the thought of SecureML~\cite{SECUREML} and ABY3~\cite{ABY3}, and design a more efficient protocol for fixed-point multiplication.

We define a 2-out-of-4 secret sharing, denoted as $\binom{4}{2}$-sharing, to enable efficient arithmetics.
And we consider all the shares are in a $\mathbb{Z}_{2^n}$ ring.
Concretely, to share an integer $x \mod 2^n$, we encode it as a tuple of shares: $\share{x} = (x_1, x_1', x_2, x_2', x_a, x_a', x_b, x_b')$.
\server{1} holds $(x_1, x_1')$ where $x_1$ and $x_1'$ are two independent random integers; 
\server{2} holds $(x_2, x_2')$ where $x_2 = x - x_1$ and $x_2' = x - x_1'$; 
\server{a} holds $(x_a, x_a')$ where $x_a = x_2$ and $x_a'$ = $x_1'$;
\server{b} holds $(x_b, x_b')$ where $x_b = x_1$ and $x_b'$ = $x_2'$.
It can be easily seen that the two numbers each server holds are independently uniformlly distributed in the ring $\mathbb{Z}_{2^n}$, thus the non-colluding servers learns nothing about $x$. 
Note that all the linear arithmetic operations of the secret shares, such as $+, -, \times$, are over the ring $\mathbb{Z}_{2^n}$.
The division operation $x / 2^d$ stands for shifting the bits of $x$ right in the two's-complement representation.


\para{Sharing initialization. }
To share a number $x$, a client encodes it to $x_1$ and $x_2$ where $x_1$ is randomly sampled in $\mathbb{Z}_{2^n}$ and $x_2 = x - x_1$.
Then it sends $x_1$ and $x_2$ to \server{1} and \server{2} respectively.
After receiving the shares, \server{1} sends $x_1$ to \server{b}, and \server{2} sends $x_2$ to \server{a}. Thus the four servers holds $x_1, x_2, x_a, x_b$ respectively.
Meanwhile, \server{1} and \server{2} generate a random number $r$ using the same seed, and calculate $x_1' = x_1 - r$ and $x_2' = x_2 + r$ respectively.
Finally, \server{1} sends $x_1'$ to \server{a} while \server{2} sends $x_2'$ to \server{b}, and the servers get the $\binom{4}{2}$-sharing of $x$.
In this situation, each server only sees two independent random integers that are uniformly distributed in $\mathbb{Z}_{2^n}$, and no information about $x$ is revealed to each server.

\para{Free addition. } It can be easily seen that the $\binom{4}{2}$-sharing over $\mathbb{Z}_{2^n}$ is additively homomorphic, i.e., $\share{x} + \share{y} = \share{x + y}$, and the result shares still satisfy the above constraints.
Thus each server can locally compute the share of the sum.
Similarly, for shared bits over $\mathbb{Z}_{2}$, the \code{XOR} operation is also free.

\para{Supporting decimals. }
We map a decimal $x$ to $\mathbb{Z}_{2^n}$ as a fixed-point number: we scale it by a factor of $2^d$ and represent the rounded integer $\map{x} = \lfloor 2^dx \rfloor$ as a $n$-bit integer using the two's-complement reprensentation.
This mapping scheme is commonly used (see \cite{SECUREML, ABY3, FIXED_POINT, IMPROVED_FIXED}, and the precision is $2^{-d}$.
It is obvious that, while ignoring the precision loss, $\share{\map{x}}$ remains additively homomorphic.

\subsection{Fixed-point multiplication. }
\label{section:multiplication}
To support efficient fixed-point multiplication, we combine the thought of SecureML~\cite{SECUREML} and ABY3~\cite{ABY3} which are also secret-sharing based approaches. 
But compared with SecureML and ABY3, our fixed-point multiplication does not require precomputation and only needs 1 round of communication, and preserves the same online communication complexity for each server when ultilizing \emph{fully-duplex} communication, as Protocol~\ref{protocol:multiplication} shows.
A remarkable thing is that, each pair of servers share a random string and use the string as the seed of a pseudorandom function, like \cite{IMPROVED_3PC, ABY3, GC_3PC} do, thus they can get a same random number without communication. The thought of this protocol is from \cite{SECUREML, ABY3, IMPROVED_3PC} and the security is similar with them. Thus we omit the proof detail here due to space limitation. To argue the correctness, we observe that
\begin{align*}
    & 2^dz_1 + 2^dz_2 \\ 
    &= (x_1y_1' - r_{12} + x_by_b' + r_{ab}) + (x_2y_2' + r_{12} + x_ay_a' - r_{ab}) \\
    &= x_1y_1' + x_by_b' + x_2y_2' + x_2y_a' \\
    &= x_1y_1' + x_1y_2' + x_2y_2' + x_2y_1' \\
    &= (x_1 + x_2)(y_1' + y_2') = xy
\end{align*}
This means that $(2^dz_1, 2^dz_2)$, namely $(t_1 + t_b, t_2 + t_a)$, is a two-party share of $xy$. 
Thus, according to the theorem in \cite{SECUREML}, $(z_1, z_2)$ is a two-party share of $xy/2^d$  i.e. $z_1 + z_2 = xy/2^d$ with extremely high probability.
The same applies to $z_1'$ and $z_2'$, i.e. $z_1' + z_2' = xy/2^d$.
Also, we can see that $z_1 = z_b$, $z_2 = z_a$, $z_1' = z_a'$ and $z_2' = z_b'$. 
This means that the result shares of $z$ still satisfy the constraint of $\binom{4}{2}$-sharing.

\begin{protocol}[tb]
\DontPrintSemicolon
\KwInput{Shares of two fixed-point values $\share{x}$ and $\share{y}$}
\KwOutput{$\share{z}$ where $z = xy / 2^d$}
\KwStep{
\begin{enumerate}[a)]
    \item \server{1} generates $r_{12}$ and $r_{12}'$ and \\
                     calculates $t_1 = x_1y_1' - r_{12}$ and $t_1' = x_1'y_1 - r_{12}'$. \\ 
                     Then it sends $t_1$ to \server{b} and sends $t_1'$ to \server{a}.
    \item \server{2} generates $r_{12}$ and $r_{12}'$ and \\
                     calculates $t_2 = x_2y_2' + r_{12}$ and $t_2' = x_2'y_2 + r_{12}'$. \\ 
                     Then it sends $t_2$ to \server{a} and sends $t_2'$ to \server{b}.
    \item \server{a} generates $r_{ab}$ and $r_{ab}'$ and \\
                     calculates $t_a = x_ay_a' - r_{ab}$ and $t_a' = x_a'y_a - r_{ab}'$. \\ 
                     Then it sends $t_a$ to \server{2} and sends $t_a'$ to \server{1}.
    \item \server{b} generates $r_{ab}$ and $r_{ab}'$ and \\
                     calculates $t_b = x_by_b' + r_{ab}$ and $t_b' = x_b'y_b + r_{ab}'$. \\ 
                     Then it sends $t_b$ to \server{1} and sends $t_b'$ to \server{2}.
    \item \server{1} sets $z_1 = (t_1 + t_b) / 2^d$ and $z_1' = (t_1' + t_a') / 2^d$; \\
          \server{2} sets $z_2 = (t_2 + t_a) / 2^d$ and $z_2' = (t_2' + t_b') / 2^d$; \\
          \server{a} sets $z_a = (t_a + t_2) / 2^d$ and $z_a' = (t_a' + t_1') / 2^d$; \\
          \server{b} sets $z_b = (t_b + t_1) / 2^d$ and $z_b' = (t_b' + t_2') / 2^d$.
\end{enumerate}
}
\caption{Fixed-point multiplication protocol.}
\label{protocol:multiplication}
\end{protocol}

\subsection{Bit extraction}
Although machine learning tasks seldomly employ bitwise operations directly, some non-linear operations such as comparison can be evaluated efficiently by implementing them using bit extraction: to compare two shared numbers $\share{x}$ and $\share{y}$ in the ring $\mathbb{Z}_{2^n}$, we first calcuate $\share{z} = \share{x} - \share{y}$ and then extract the bit share $\share{c}^B$ where $c$ is the most significant bit of $z$. As we use two's-complement reprensentation, $c = 1$ means $z$ is negative (i.e. $x < y$).
Garbled circuit~\cite{YAO_PROTOCOL}, though is efficient for boolean circuits (e.g. comparison) and requires constant rounds, usually causes high communication cost, especially when the data size is large and the throughput is the main concern, which is the case of machine learning tasks.

A direct way of bit extraction is to first convert each share (i.e. $x_1, x_1', x_2, \dots$) of $x$ to bit representation and then perform addition circuits to get the bit representation of $x$, like \cite{ABY3} does.
However, with $\binom{4}{2}$ sharing, each server should send or receive at least 1 message to convert each share to the bit representation.

To optimize the bit extraction, we observe that it is not necessary to evaluate the whole addition circuit, as we only need to get the final bit and there is no need to output some intermediate results. Protocol~\ref{protocol:bit_extraction} shows this. $x[i]$ means the $i$-th bit of $x$, while $x[1:k]$ means the first $k$ bits of $x$.
To extract the $k$-th bit of $x$, we regard $x_1$ and $x_2$ as two bit arrays, each of which consists of $n$ bits, and perform an addition circuit on them, then we use the 1-bit adder in \cite{IMPROVED_GC} to calculate the carry bits: $c[i + 1] = (x_1[i] \oplus c[i]) \wedge (x_2[i] \oplus c[i]) \oplus c[i]$. 
To see how Protocol~\ref{protocol:bit_extraction} gets $c[i+1]$, we first see from step $a$-$c$ that $u_2 \oplus u_a = x_1[1:k] \wedge x_2[1:k]$.
Then in step $h$, we get 
\begin{align*}
    &  c[i+1]_1' \oplus c[i+1]_2' \\
    &= ((x_1[i] \wedge c[i]_1' \oplus b_{12}) \oplus (x_a[i] \wedge c[i]_a' \oplus u_a[i] \oplus b_{ab})) \oplus \\ 
      &\ \ \ \ \  ((x_2[i] \wedge c[i]_2' \oplus u_2[i] \oplus b_{12}) \oplus (x_b[i] \wedge c[i]_b' \oplus b_{ab})) \\
    &= (x_1[i] \wedge c[i]') \oplus (x_2[i] \wedge c[i]') \oplus (x_1[i] \wedge x_2[i]) \\ 
    &= (x_1[i] \oplus c[i]) \wedge (x_2[i] \oplus c[i]) \oplus c[i] = c[i+1]
\end{align*}
Note that in step $d$-$h$, we only use $c_1', c_2', c_a'$ and $c_b'$. Thus we can get each carry bit without calculating $c_1, c_2, c_a$ and $c_b$.
Finally, step $i$-$m$ calculate the output bit as $c_[k] = c[k-1] \oplus x_1[k-1] \oplus x_2[k-1]$ and reconstruct the shares to $\binom{4}{2}$-sharing.

The security of this protocol is similar with Protocol~\ref{protocol:multiplication}, as each server only receives independent random bits at each step.
Thus if we perform the addition circuit by chaining the 1-bit adders together, each server sends no more than $2k$ bits in total, which is the same as the bit extraction protocol in \cite{ABY3}.
We can further optimize the communiication complexity by letting \server{1} and \server{b} share a half of the shares of $x_1[1:k] \wedge x_2[1:k]$ in step $a$-$c$. Thus each server only needs to send $1.5k$ bits in total.
Another optimization we can use is to employ a parallel prefix adder (PPA)~\cite{PPA} which uses a divide and conquer strategy and reduces the total number of rounds to $O(\log k)$, like \cite{ABY3} does.

\begin{protocol}[tb]
\DontPrintSemicolon
\KwInput{$\share{x}$ and $k$}
\KwOutput{$\share{c}$ where $c$ is $k$-th bit of $x$}
\KwStep{
Each server initilizes $c$ as 0. Then run as follows:
\begin{enumerate}[a)]
    \item \server{1} generates $r_{12}$ in $\mathbb{Z}_{2^k}$, and calcuates \\ $u_1 = x_1[1:k] \oplus r_{12}$. 
        Then it sends $u_1$ to \server{a}. 
    \item \server{2} generates $r_{12}$ in $\mathbb{Z}_{2^k}$ and sets $u_2 = x_2[1:k] \wedge r_{12}$.
    \item \server{a} sets $u_a = u_1 \wedge x_a$.
\end{enumerate}
For $(i = 1; i <= k - 1; i = i + 1) \{$
\begin{enumerate}[a)]
    \item \server{1} generates a random bit $b_{12}$ and calculates $t_1' = x_1[i] \wedge c_1' \oplus b_{12}$. 
        Then it sends $t_1'$ to \server{a}.
    \item \server{2} generates a random bit $b_{12}$ and calculates $t_2' = x_2[i] \wedge c_2' \oplus u_2[i] \oplus b_{12}$. 
        Then it sends $t_2'$ to \server{b}.
    \item \server{a} generates a random bit $b_{ab}$ and calculates $t_a' = x_a[i] \wedge c_a' \oplus u_a[i] \oplus b_{ab}$. 
        Then it sends $t_a'$ to \server{1}.
    \item \server{b} generates a random bit $b_{ab}$ and calculates $t_b' = x_b[i] \wedge c_b' \oplus b_{ab}$. 
        Then it sends $t_b'$ to \server{2}.
    \item \server{1} sets $c_1' = t_1' \oplus t_a'$; 
          \server{2} sets $c_2' = t_2' \oplus t_b'$; 
          \server{a} sets $c_a' = t_a' \oplus t_1'$; 
          \server{b} sets $c_b' = t_b' \oplus t_2'$. 
\end{enumerate}
$\}$
\begin{enumerate}[a)]
    \item \server{1} generates a random bit $b_{12}$ and calculates $c_1' = x_1[k] \oplus c_1' \oplus b_{12}$. 
        Then it sends $c_1'$ to \server{a}.
    \item \server{2} generates a random bit $b_{12}$ and calculates $c_2' = x_2[k] \oplus c_2' \oplus b_{12}$. 
        Then it sends $c_2'$ to \server{b}.
    \item \server{a} generates a random bit $b_{ab}$ and calculates $c_a = x_a[k] \oplus c_a' \oplus b_{ab}$. 
        Then it sends $c_a$ to \server{2}.
    \item \server{b} generates a random bit $b_{ab}$ and calculates $c_b = x_b[k] \oplus c_b' \oplus b_{ab}$. 
        Then it sends $c_b$ to \server{1}.
    \item \server{1} sets $c_1 = c_b$; 
          \server{2} sets $c_2 = c_a$;
          \server{a} sets $c_a' = c_1'$; \\
          \server{b} sets $c_b' = c_2'$;
\end{enumerate}
}
\caption{Bit extraction protocol.}
\label{protocol:bit_extraction}
\end{protocol}

\subsection{OT protocol}
Oblivious transfer (OT)~\cite{OT} enables oblivious selection between two numbers without revealing the index and the private numbers: given a shared bit $\share{c}^B$ and two shared numbers $\share{x}$ and $\share{y}$, OT outputs $\share{x}$ if $c=1$, and outputs $\share{y}$ otherwise.
Note that this statement is equivalent to another one: given a shared bit $\share{c}^B$ and a shared number $\share{\delta}$, OT outputs $\share{x}$ if $c=1$, and outputs $\share{0}$ otherwise. This is because we can define $\delta = y - x$ in the first statement and add the OT output to $\share{x}$.

With efficient OT protocols, many common machine learning components, such as ReLu and piecewise functions, can be evaluated efficiently.
Although the goal is to output $cx$, we cannot perform OT using Protocol~\ref{protocol:multiplication} directly, as the bit share is over $\mathbb{Z}_2$ while the integer share is over $\mathbb{Z}_{2^n}$.
To construct our OT protocol using $\binom{4}{2}$-sharing, we first convert the \code{XOR} opeartion into addition over $\mathbb{Z}_{2^n}$, then perform arithmetical computation on the shares. Specifically, we have:
\begin{align*}
    cx &= (c_1' \oplus c_2')(x_1 + x_2) \\
       &= (c_1' + c_2' - 2c_1'c_2')(x_1 + x_2) \\ 
       &= c_1'x_1 + (1-2c_1')c_2'x_1 + c_2'x_2 + (1-2c_2')c_1'x_2 \\
       &= c_1'x_1 + (1-2c_1')c_b'x_b + c_2'x_2 + (1-2c_2')c_a'x_a
\end{align*}
In the above formulation, $c_1'x_1, c_2'x_2, c_a'x_a$ and $c_b'x_b$ can be computed locally, thus the servers only need to introduce $(1-2c_1')$ and $(1-2c_2')$ into the result.
Protocol~\ref{protocol:ot} shows our 4-party OT scheme.
The analysis of correctness and security are similar as above. The protocol runs in 1 round and each server needs to send 4 messages.

\begin{protocol}[tb]
\DontPrintSemicolon
\KwInput{$\share{x}$ over $\mathbb{Z}_{2^n}$ and $\share{c}$ over $\mathbb{Z}_2$}
\KwOutput{$\share{y}$ where $y = x$ if $c = 1$, and $y=0$ otherwise. }
\KwStep{
\begin{enumerate}[a)]
    \item \server{1} generates random numbers $r_{12}, r_{1b}, r_{12}'$ and $r_{1a}'$ in $\mathbb{Z}_{2^n}$, and calculates 
        $t_1 = c_1'x_1 - r_{12}, t_1' = c_1x_1' - r_{12}', e_1 = (1 - 2c_1')r_{12} + r_{1b}, e_1' = (1 - 2c_1)r_{12}' + r_{1a}'$. \\
        Then it sends $t_1'$ and $e_1$ to \server{a}, and sends $t_1$ and $e_1'$ to \server{b}.
    \item \server{2} generates random numbers $r_{12}, r_{2a}, r_{12}'$ and $r_{2b}'$ in $\mathbb{Z}_{2^n}$, and calculates 
        $t_2 = c_2'x_2 - r_{12}, t_2' = c_2x_2' - r_{12}', e_2 = (1 - 2c_2')r_{12} + r_{2a}, e_2' = (1 - 2c_2)r_{12}' + r_{2b}'$. \\
        Then it sends $t_2'$ and $e_2$ to \server{b}, and sends $t_2$ and $e_2'$ to \server{a}.
    \item \server{a} generates random numbers $r_{ab}, r_{2a}, r_{ab}'$ and $r_{1a}'$ in $\mathbb{Z}_{2^n}$, and calculates 
        $t_a = c_a'x_a - r_{ab}, t_a' = c_ax_a' - r_{ab}', e_a = (1 - 2c_a')r_{ab} + r_{2a}, e_a' = (1 - 2c_a)r_{ab}' + r_{1a}'$. \\
        Then it sends $t_a'$ and $e_a$ to \server{1}, and sends $t_a$ and $e_a'$ to \server{2}.
    \item \server{b} generates random numbers $r_{ab}, r_{1b}, r_{ab}'$ and $r_{2b}'$ in $\mathbb{Z}_{2^n}$, and calculates 
        $t_b = c_b'x_b - r_{ab}, t_b' = c_bx_b' - r_{ab}', e_b = (1 - 2c_b')r_{ab} + r_{1b}, e_b' = (1 - 2c_b)r_{ab}' + r_{2b}'$. \\
        Then it sends $t_b'$ and $e_b$ to \server{2}, and sends $t_b$ and $e_b'$ to \server{1}.
    \item \server{1} sets $y_1 = (1 - 2c_1')t_b + c_1'x_1 + e_a - r_{1b}$ and $y_1' = (1 - 2c_1)t_a' + c_1x_1' + e_b' - r_{1a}'$; \\
          \server{2} sets $y_2 = (1 - 2c_2')t_a + c_2'x_2 + e_b - r_{2a}$ and $y_2' = (1 - 2c_2)t_b' + c_2x_2' + e_a' - r_{2b}'$; \\
          \server{a} sets $y_a = (1 - 2c_a')t_2 + c_a'x_a + e_1 - r_{2a}$ and $y_a' = (1 - 2c_a)t_1' + c_ax_a' + e_2' - r_{1a}'$; \\
          \server{b} sets $y_b = (1 - 2c_b')t_1 + c_b'x_b + e_2 - r_{1b}$ and $y_b' = (1 - 2c_b)t_2' + c_bx_b' + e_1' - r_{2b}'$;
\end{enumerate}
}
\caption{4-party OT protocol.}
\label{protocol:ot}
\end{protocol}

\subsection{Other basic POs}
Limited by space, we only briefly introduce other POs here.  
We implement the \emph{comparison} operation using oblivious transfer which enables secure selection between two private numbers and bit extraction which extracts a specific bit of a private number, in a similar way as \cite{ABY3}, and thus the correctness and security directly follow \cite{ABY3}.
We also implement two basic \emph{bitwise operations} \code{XOR} and \code{AND}, and we can get all kinds of bitwise operations by composing the two ones. Actually, \code{XOR} is the addition modulo 2, while \code{AND} is the multiplication modulo 2. Thus if we use use $\binom{4}{2}$-sharing to represent a bit, we can implement \code{XOR} and \code{AND} similar to addition and multiplication of integer operations in $\mathbb{Z}_{2^n}$.

\subsection{Derived POs}
\label{subsection:other_operation}
We can compose multiple basic POs and form more complex derived POs commonly used in machine learning algorithms.
For example, to compute a \emph{ReLu} function $f(x) = \max(0, x)$ which is commonly used as a activation function in neural networks, we can first extract the most significant bit of $-x$ (which indicates $x$ is positive or not), then use the OT protocol to get $f(x)$.
For \emph{division}, we can use the Newton-Raphson algorithm~\cite{NEWTON_RAPSHON} to approximate the result.
Similarly, to implement the logistic function $f(x) = \frac{1}{1+e^{-x}}$, we can use the Euler method~\cite{EULER_METHOD}. Another alternative implementation for the logistic function is the piecewise function in \cite{SECUREML, ABY3}.
We also implement other common maths functions using similar numerical methods, such as \emph{sqrt}, \emph{log}, \emph{exp} and \emph{max\_pooling} in a similar way.
With these basic POs and derived POs, we can further implement complex algorithms as usual.

\subsection{POs for performance optimization}
\label{section:building_block_optimization}
\label{subsection:optimization_batch}

We provide the following three sets of POs whose functionality is already covered by the basic POs, but the separate versions can significantly improve performance in certain cases.  Programmers can use these POs directly. 

\para{Batching POs.}
\emph{Batch up} is a commonly used optimization in MPC frameworks~\cite{SHAREMIND, SPDZ, PICCO}, which batches up independent data transfers among the servers and thus reduce the fixed overhead.
Array POs natively support batch up. 
And as many machine learning algorithms heavily utilize array operations, this optimization reduces communication rounds and can improve performance significantly. 

\para{Multiply by public variables.  }
In a case where an operation involves both public and private variables, we can optimize performance by revealing the public variables.
Multiplication benefits from the optimization the most, as
the servers only need to multiply their shares by the public variables directly and there is no necessary communication.

\para{Dot and outer product. }
Dot and outer product of matrices are frequently used in common machine learning algorithms.
For example, logistic regression and neural networks use dot product for forward propagation, represented as $Y = W \cdot X + b$.
Outer product is often used for calculating gradients.
While implementing them using for-loops, there are too many duplicated transfers for each element, as each element will be multiplied by several other elements in a multi-dimensional situation.
We thus provide built-in optimized dot and outer product. Specifically, for two private matrices $\share{A}$ and $\share{B}$, we can calculate the dot product as $\share{A} \cdot \share{B} = A_1 \cdot B_1' + A_2 \cdot B_2' + A_a \cdot B_a' + A_b \cdot B_b'$. 
This optimization significantly reduces communication cost. As an example, given two $n \times n$ matrices, a for-loop for dot product triggers $n^3$ multiplications and the communication complexity is $O(n^3)$, while the optimized one only incurs communication complexity of $O(n^2)$.

\section{Front-end and Optimizations}
\label{section:optimizations}

We now introduce the design and implementation of the programming interfaces.  
\alert{Our goal is to provide intuitive interfaces and automatic optimizations to avoid steep learning curves and enable programmers to focus on the machine learning algorithm itself.}

\subsection{PrivPy Front-end Features}

A \sysname program is a valid Python program with NumPy-style data type definitions. We use three real code segments to illustrate the PrivPy features essential to implementing data mining algorithms.  

Fig.~\ref{fig:code_example} shows a \sysname program that computes the logistic function $f(x) = 1/(1+e^{-x})$ using the Euler method~\cite{EULER_METHOD}.  Fig.~\ref{fig:sample_code_mf} shows an extra example of matrix factorization, which decomposes a large private matrix $x$ to two latent matrices $P$ and $Q$.  Lastly, Fig.~\ref{fig:sample_code_nn} shows an example of neural network inference.

\begin{figure}
    \begin{lstlisting}[language=Python]
    x = privpy.ss(clientID)
    def logistic(x, start, iter_cnt): 
      result = 1.0 / (1 + math.exp(-start)) 
      deltaX = (x - start) / iter_cnt 
      for i in range(iter_cnt):
        derivate = result * (1 - result) 
        result += deltaX * derivate 
      return result
    result = logistic(x, 0, 100) # main()
    result.reveal()
    \end{lstlisting}
    \caption{Example \sysname code: logistic function.}
    \label{fig:code_example}
\end{figure}

\begin{figure}[tb]
	\begin{lstlisting}[language=Python, numbers=none]
	import privpy as pp
	x = ... # read data using ss()
	factor,gamma,lamb,iter_cnt = initPublicParameters()
	n,d = x.shape
	P = pp.random.random((n,factor))
	Q = pp.random.random((d,factor))
	for _  in range(iter_cnt):
	  e = x - pp.dot(P,pp.transpose(Q))
	  P1 = pp.reshape(pp.repeat(P,d,axis=0),P.shape[:-1] + (d,P.shape[-1]))
	  e1 = pp.reshape(pp.repeat(e,factor,axis=1),e.shape + (factor,))
	  Q1 = pp.reshape(pp.tile(Q,(n,1)),(n,d,factor))
	  Q += pp.sum(gamma * (e1 * P1 - lamb * Q1),axis = 0)/n
	  Q1 = pp.reshape(pp.tile(Q,(n,1)),(n,d,factor))
	  P += pp.sum(gamma * (e1 * Q1 - lamb * P1),axis = 1)/d
	P.reveal(); Q.reveal()
	\end{lstlisting}
	\caption{Example \sysname code: matrix factorization.}
	\label{fig:sample_code_mf}
\end{figure}

\begin{figure}[tb]
	\begin{lstlisting}[language=Python, numbers=none]
	import privpy as pp
	x = ... # read data using ss()
	W, b = ... # read model using ss()
	for i in range(len(W)):
	  x = pp.dot(W.T, x) + b
	  x = pp.relu(x)
	res = pp.argmax(x, axis=1)
	res.reveal()
	\end{lstlisting}
	\caption{Example \sysname code: neural network inference.}
	\label{fig:sample_code_nn}
\end{figure}

\para{Basic semantics.  }
Unlike many domain-specific front-ends~\cite{TASTY, L1, SHAREMIND}, which require the programmers to have knowledge about cryptography and use customized languages, 
the program itself (lines 2-9) is a plain Python program, which can run in a raw Python environment with cleartext input, and the user only needs to add two things to make it private-preserving in \sysname: 
$(i)$ Declaring the private variables.
Line 1 declares a private variable \code{x} as the input from the client $clientID$ using the \code{ss} function. 
$(ii)$ Getting results back.  
The function \code{reveal} in line 10 allows clients to recover the cleartext of the private variable.
Programmers not familiar with cryptography, such as machine learning programmers, can thus implement algorithms with minimal effort.

\para{All operations support both scalar and array types.  }
\sysname supports scalars, as well as arrays of any shape.
\alert{Supporting array operations is essential for writing and optimizing machine learning algorithms which rely heavily on arrays.}
While invoking the \code{ss} method, \sysname detects the type and the shape of \code{x} automatically.
If \code{x} is an array, the program returns an array of the same shape, containing the function on every element in \code{x}.
Following the NumPy~\cite{NUMPY} semantics, we also provide \emph{broadcasting} that allows operations between a scalar and an array, as well as between arrays of different shapes, two widely used idioms. 
That is why the \code{logistic} function in Fig.~\ref{fig:code_example} works correctly even when \code{x} is a private array. 
As far as we know, existing MPC front-ends, such as \cite{SHAREMIND, L1, PICCO, SPDZ}, do not support such elegant program. 
For example, PICCO~\cite{PICCO} only supports operations for arrays of equal shape.


\para{Private array types.  }
Array operations are pretty common in machine learning algorithms.
The private array class in \sysname encapsulates arrays of any shape.  Users only need to pass a private array to the constructor, then the constructor automatically detects the shape. Like the array type in Numpy~\cite{NUMPY}, our private array supports \emph{broadcasting}, i.e. \sysname can handle arithmetic operations with arrays of different shapes by ``broadcasting'' the smaller arrays (see \cite{BROADCASTING} for details). 
For example, given a scalar $x$, a $4 \times 3$ array $A$, a $2 \times 4 \times 3$ array $B$ and a $2 \times 1 \times 3$ array $C$, the expressions $x \bigodot A$, $A \bigodot B$ and $B \bigodot C$ are all legal in \sysname, where $\bigodot$ can be $+, \times \text{and} >$ etc. Note that in \sysname, the above variables can be either public or private.
With broadcasting, programmers can write elegant machine learning algorithms regardless of the shapes of the inputs and model parameters.

We also implement most of the ndarray methods of Numpy, with which application programmers can manipulate arrays conveniently and efficiently, except for the methods related with IO (we leave IO as the future work). Table~\ref{table:operations} lists the ndarray methods we have implemented (see \cite{NDARRAY} for details of \code{numpy.ndarray}).

Broadcasting and ndarray methods are essential for implementing common machine learning algorithms which usually handle arrays of different shapes.

Both Fig.~\ref{fig:sample_code_mf} and Fig.~\ref{fig:sample_code_nn} demonstrate \emph{ndarray} methods in \sysname. Users can implement the algorithms in plain Python, then just replace the \emph{Numpy} package with \emph{\sysname} package and add private variables declaration. Actually, by replacing all \code{privpy} with \code{numpy}, the main parts of Fig.~\ref{fig:sample_code_mf} and Fig.~\ref{fig:sample_code_nn} can run directly in raw Python environment with cleartext inputs. 

\para{Support for large arrays.  }
Mapping the data onto secret shares unavoidably increases the data size. Thus, real-world datasets that fit in memory in cleartext may fail to load in the private version.  For example, the $1,000,000 \times 5,048$ matrix require over 150GB memory.  

\para{Automatic code rewriting. }
With the program written by users, the interpreter of our front-end parses it to basic privacy-preserving operations supported by the back-end, and the optimizer automatically rewrites the program to improve efficiency (see \sect{section:optimizations} for details). 
\alert{This optimization can help programmers avoid performance ``pit falls'' in MPC situation.}

\subsection{Implementations}

\para{Based on Plain Python Interpreter. }
We write our backend in C++ for performance, and we implement our frontend in Python to keep python compatibility.  The backend is linked to the frontend as a library on each server to reduce the overhead between the frontend and backend.
During a execution task, the same Python code is interpretered on each server and client in parallel.

\para{NumPy-style data type definitions and operator overloading.  } We define our own data types \code{SNum} and \code{SArr}, to represent the secret numbers and arrays, respectively.  Then we overload operators for private data classes, so standard operators such as $+, -, *, >, =$ work on both private and public data.  The implementation of these overloaded operators chooses the right POs to use based on data types and the sizes at runtime.  

\para{Automatic disk-backed large arrays.  }
We provide a \code{LargeArray} class that transparently uses disks as the back storage for arrays too large to fit in memory. 

\subsection{Code analysis and optimization}
\label{subsection:optimization_ast}

Comparing to the computation on cleartext, private operations have very distinct cost, and many familiar programming constructs may lead to bad performance, creating ``performance pitfalls''.  Thus, we provide aggressive code analysis and rewriting to help avoid these pitfalls.  For example, it is fine to write an element-wise multiplication of two vectors in plain Python program.
\begin{lstlisting}[language=Python,numbers=none]
for i in range(n): z[i] = x[i] * y[i]
\end{lstlisting}

However, this is a typical anti-pattern causing performance overhead due to the $n$ multiplications involved, comparing to a single array operation (\sect{subsection:optimization_batch}).  
To solve the problem, we build a source code analyzer and optimizer based on Python's abstract syntax tree (AST) package~\cite{AST}.  Before the servers execute the user code, our analyzer scans the AST and rewrites anti-patterns into more efficient ones.  
In this paper, we implement three examples: 

\para{For-loops vectorization.  }
Vectorization~\cite{VECTORIZATION} is a well-known complier optimization. 
This analyzer rewrites the above for-loop into a vector form $\vec{z} = \vec{x} * \vec{y} $. 
The rewriter also generates code to initialize the vector variables. 


\para{Common factor extraction.  }
We convert expressions with pattern
$x * y_1 + x * y_2 + \dots + x * y_n$
to  
$x * (y_1 + y_2 + \dots + y_n)$. 
In this way, we reduce the number of $\times$ from $n$ to 1, saving significant communication time.


\para{Common expression vectorization.  }
Programmers often write vector expressions explicitly, like $x_1*y_1 + x_2*y_2 + \dots + x_n*y_n$, especially for short vectors.  The optimizer extracts two vectors $\vec{x} = (x_1, x_2, \dots, x_n)$ and $\vec{y} = (y_1, y_2, \dots, y_n)$, and rewrite the expression into a vector dot product of $\vec{x}\cdot \vec{y}$.
Note that $x_1, x_2, \dots, x_n$ do not have to be the same shape, as \sysname supports batch operations with mixed shapes.


\para{Reject for unsupported statements.  }
We allow users to write legal Python code that we cannot run correctly, such as branches with private conditions (actually, most MPC tools do not support private conditions~\cite{PICCO, OBLIVM}, or only support limited scenarios~\cite{PICCO, OBLIV_C}). 
In order to minimize users' surprises at runtime, we perform AST-level static checking, then reject unsupported statements at the initialization phase and terminate with an error.

\section{Evaluation}
\label{section:experiments}

\para{Testbed.}
We run our experiments on four Amazon EC2 virtual machines.  All machines are of type \code{c5.2xlarge} with 8 Intel Xeon Platinum 8000-series CPU cores and 64 GB RAM.
Each machine has a 1 GB Ethernet adapter running in full-duplex mode. 
In our experiments, we consider two network settings: a LAN setting where each virtual machine has 10Gbps incoming and outgoing bandwidth, and a WAN setting where the bandwidth of each virtual machine is 50 Mbps and the RTT latency is 100 ms.

\para{Parameter setting. }
All arithmetic shares are over $\mathbb{Z}_{2^{128}}$, and we set $d = 40$, which means the scaling factor is $2^{40}$. 
We repeat each experiment 10 times and report the average values.

\para{\sysname implementation.  }
We implement the front-end of \sysname with Python, and use C++ to implement our computation engine. 
And we use the built-in \code{\_\_int128} type of \code{gcc} to implement 128-bit integers.
We compile the C++ code using \code{g++ -O3}, and wrap it into Python code using the Boost.Python library~\cite{BOOST}.  We use SSL with 1024-bit keys to protect all communications. 

\subsection{Microbenchmarks}
We first perform microbenchmarks in the LAN setting to show the performance of basic operations and the benefits of optimizations.

\para{Basic operations. }
\sysname engine supports efficient fundamental operations, including addition, fixed-point multiplication and comparison. Addition can be done locally, while multiplication and comparison involve communication. Thus we demonstrate the performance of the latter two, as Table~\ref{table:time-op} shows.

\begin{table}
\centering
\small
\begin{tabular}{|c|c|}
\hline
\tabincell{c}{fixed-point multiplication} & comparison \\
\hline
$10,473,532$ & $128,2027$ \\
\hline
\end{tabular}
\caption{Throughput (ops/second) of fundamental operations over $\mathbb{Z}_{2^{128}}$ in the LAN setting.}
    \vspace{-0.15in}
\label{table:time-op}
\end{table}

\para{Client-server interaction. }
We evaluate the performance of the secret sharing process $ss$, with which the clients split raw data to secret shares and send them to the servers, and the reverse process $reveal$, with which the clients receive the shares from the servers and recover them to the plaintexts. We evaluate the time (including computation and communication) with different numbers of clients and dimensions, assuming that each client holds an accordingly dimensional vector. Figure~\ref{figure:time-client} shows that even with 1000 clients and 1000-dimension vectors, it takes only less 0.3 seconds for the servers to collect/reveal all the data from/to all the clients. 

\begin{figure}
	\centering
	\includegraphics[width = 0.48\textwidth]{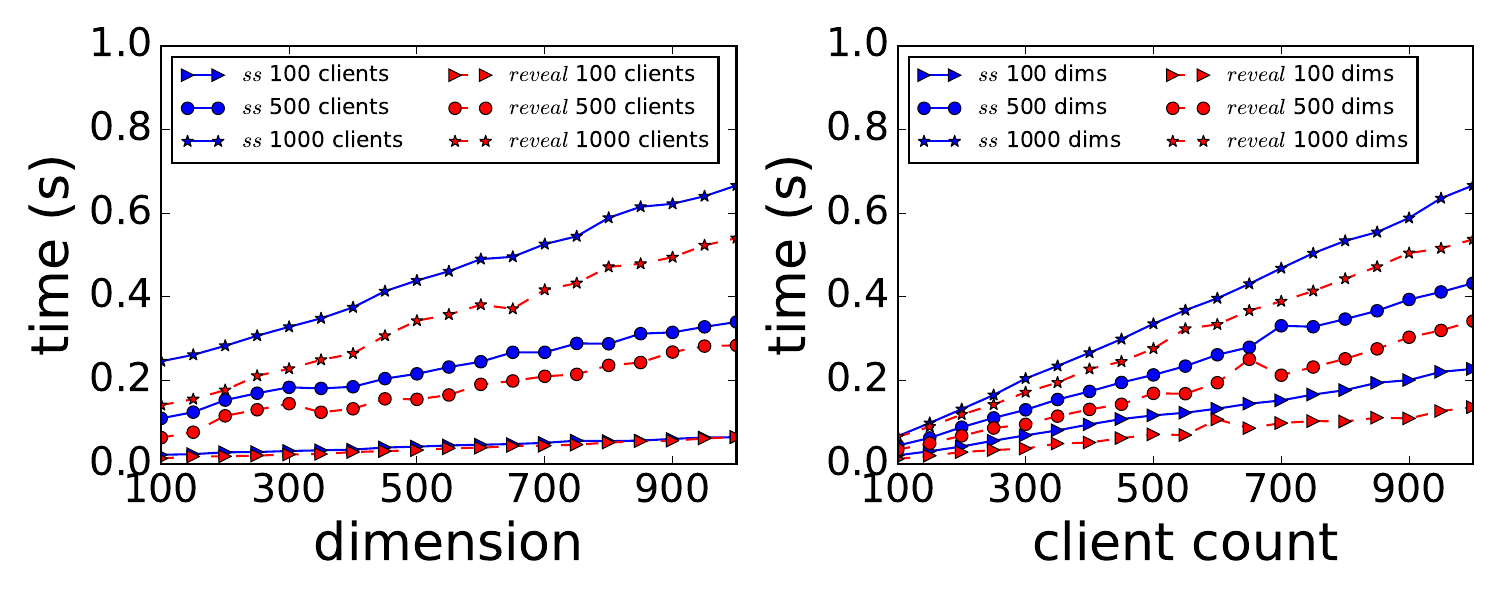}
	\caption{Performance of $ss$ and $reveal$.}
	\label{figure:time-client}
    \vspace{-0.1in}
\end{figure}

\para{Effects of batch operations.  } 
We evaluate the effectiveness of batching up, using two common operations: \emph{element-wise multiplication} and \emph{dot product} on vectors. 
For multiplication, we batch up the communication of independent operations, while for dot product of two $m$-dimensional vectors, we only need to transfer the dot-producted shares and  the communication cost is reduced from $O(m)$ to $O(1)$.
We vary the number of elements and measure the time consumption, and Fig.~\ref{figure:time-optimize} shows the result (the $y$-axis is the logarithm of time).
Both cases show benefits over $1000 \times$.

\begin{figure}
\centering
\includegraphics[width = 0.45\textwidth]{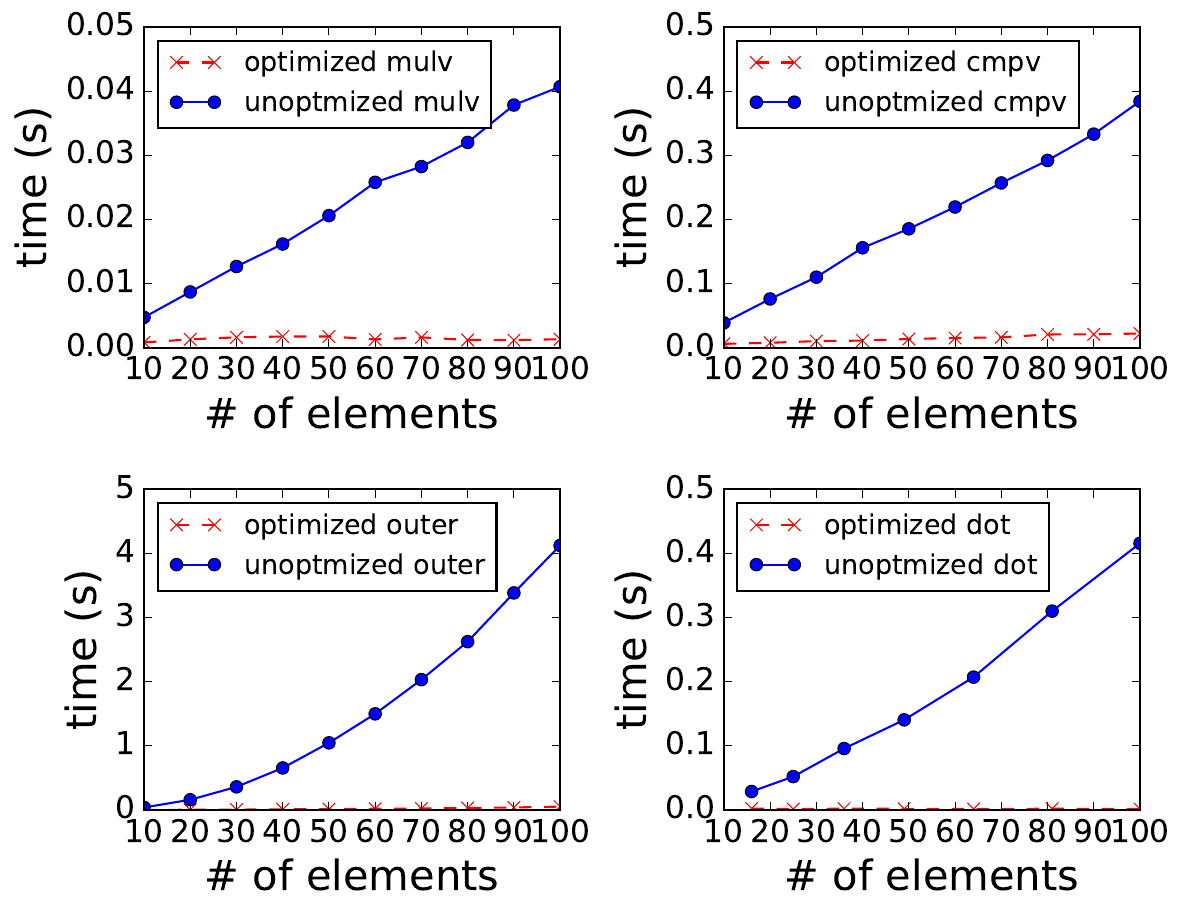}
\caption{The optimization of doing operations in batch.}
\label{figure:time-optimize}
    \vspace{-0.1in}
\end{figure}

%

\para{Effects of code optimizations.  }
We evaluate the common factor extraction and expression vectorization.  As these hand-written anti-patterns are usually small, we range the expression size from 2 to 10.  Figure~\ref{figure:time-ast} shows that more than $4\times$ performance improvement for five-term expressions in both situations. 

\begin{figure}[tb]
\centering
\includegraphics[width = 0.5\textwidth]{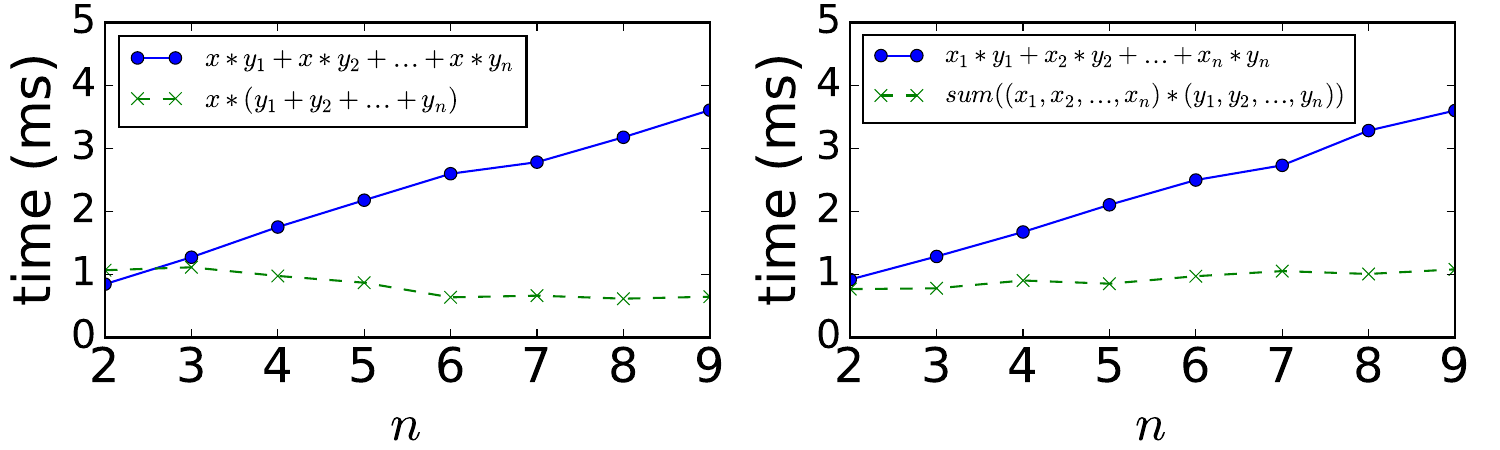}
    \vspace{-0.15in}
\caption{Code optimizer performance.}
\label{figure:time-ast}
    \vspace{-0.15in}
\end{figure}

\para{Disk-backed large array performance. }
The \sysname front-end provides a class \code{LargeArray} to automatically handle the large arrays that are too large to fit in the memory. 
As \code{LaregeClass} uses disks as back storage, we should consider the effect of the disk IO time to the overall performance.
To evaluate the performance of \code{LargeArray}, we use the Movielens dataset~\cite{MOVIELENS} which contains 1 million movie ratings from thousands of users. We encode the dataset to a $1,000,000 \times 5048$ matrix, which requires 150GB memory space in each machine. 
We then perform a dot product of a batch of elements in the dataset and a $5048$-dimensional vector.
We evaluate the performance by varying the batch size and \emph{randomly} choosing a batch of items in the dataset.
As Table~\ref{table:time_dot_large} shows, the disk IO becomes the main cost. The reason is that the program should sequentially scan the large array stored in the disk to retrieve the randomly choosen batch.

We then apply \code{LargeArray} to the training of real algorithms: 
logistic regression (LR) which is trained using SGD, and matrix factorization (MF)~\cite{MATRIX_FACTORIZATION_DP} which decomposes a large matrix into two smaller latent matrices for efficient prediction (in this paper, we decompose each $m \times n$ matrix to a $m \times 5$ matrix and a $5 \times n$ matrix).
 Table~\ref{table:time_algorithm_large} shows the result.

\begin{table}[tb]
	\centering
    \small
	\begin{tabular}{|c|c|c|c|}
		\hline
         & \tabincell{c}{communication + \\ computation} & \tabincell{c}{disk IO} & \tabincell{c}{total}\\
		\hline
        single & 0.38 & 0.7 & 1.08 \\
		\hline
        batched & 0.172 & 0.574 & 0.746 \\
        \hline
	\end{tabular}
	\caption{Time (milliseconds) for dot product of a large array.}
    \vspace{-0.15in}
	\label{table:time_dot_large}
\end{table}

\begin{table}[tb]
	\centering
    \small
	\begin{tabular}{|c|c|c|}
		\hline
        batch size & \tabincell{c}{LR} & \tabincell{c}{MF} \\
		\hline
        single & 0.027 & 0.037 \\
		\hline
        batched & 0.026 & 0.042 \\
        \hline
	\end{tabular}
	\caption{Time (seconds) for real algorithms with large arrays.}
    \vspace{-0.15in}
	\label{table:time_algorithm_large}
\end{table}

\subsection{Performance of real algorithms}
The focus of PrivPy is the \emph{algorithm performance} on real big datasets.  We present our experience in common algorithms, including \emph{logisgtic regression (LR)}, \emph{matrix factorization (MF)} and \emph{neural network (NN)}, using both ABY3 and our backend.  ABY3 engine has several configuration options, we use the most performance-optimized options of ABY3 (semi-honest assumption without precomputation) in all the evaluations.  


We perform our evaluation in both the LAN setting and the WAN setting, and use the MNIST dataset~\cite{MNIST} which includes $70,000$ labeled handwritten digits~\cite{HANDWRITTEN} with $28 \times 28$ pixels each.  
We evaluate the performance for both training and inference. 
And as our front-end supports both engines, we run \emph{the same} algorithm codes written with Python on the two engines.

Table~\ref{table:time_training} shows the average time consumed by 1 iteration of training. The logistic regression and matrix factorization is constructed as above, and the neural network (NN) has a $784$-dimensonal input layer, two $128$- hidden layers and a $10$-dimensional output layer.

For inference, we in addition evaulate the LeNet-5~\cite{LENET5} model to demonstrate convonlutional neural network (CNN). LeNet-5 has a 784-dimension input layer, 3 convolutional layers with a $5 \times 5$ kernel, 2 sum-pooling layers, 3 sigmoid layers, 1 dot product layer, 1 Radial Basis Function layer, and an \emph{argmin} function on a 10-dimension vector to get the output.
Then based on the LeNet-5 model, we add a batch normalization~\cite{BATCHNORMALIZATION} layer to each sigmoid layer to get a CNN+BN model.
The evaluation result is as Table~\ref{table:time_inference} shows.

From the evaluation results, we can see that our computation engine performs better than ABY3 for both training and inference, especially in the WAN setting.
This is because, although both ABY3 and our computation engine require no precomputation and have the same communication cost for each server, ABY3 requires 1 more round than our computation engine for fixed-point multiplication, thus causes lower performance.

\begin{table}[tb]
	\centering
    \small
	\begin{tabular}{|c|c|c|c|c|c|c|c|}
		\hline
        \multirow{2}{*}{\tabincell{c}{batch\\ size}} & \multirow{2}{*}{engine} & \multicolumn{3}{c|}{LAN} & \multicolumn{3}{c|}{WAN} \\
        \cline{3-8}
        & & LR & MF & NN & LR & MF & NN \\
		\hline
        \multirow{2}{*}{single} & ABY3 & 5.2e-3 & 7.4e-3 & 1.8e-2 & 2.16 & 0.62 & 1.27 \\
        \cline{2-8}
        & \sysname & 5.3e-3 & 7.1e-3 & 1.7e-2 & 2.61 & 0.37 & 1.16 \\
        \hline
        \multirow{2}{*}{batched} & ABY3 & 3.94 & 5.72 & 58.1 & 7.53 & 18.6 & 637 \\
        \cline{2-8}
        & \sysname & 3.92 & 5.67 & 52.5 & 7.3 & 13.2 & 554 \\
        \hline
	\end{tabular}
	\caption{Time (seconds) for training of real algorithms with different engines.}
	\label{table:time_training}
\end{table}

\begin{table}[tb]
	\centering
    \small
	\begin{tabular}{|c|c|c|c|c|c|c|c|}
		\hline
        \multirow{2}{*}{\tabincell{c}{batch\\ size}} & \multirow{2}{*}{engine} & \multicolumn{3}{c|}{LAN} & \multicolumn{3}{c|}{WAN} \\
        \cline{3-8}
        & & NN & CNN & \tabincell{c}{CNN\\+BN} & NN & CNN & \tabincell{c}{CNN\\+BN} \\
		\hline
        \multirow{2}{*}{single} & ABY3 & 1.3e-2 & 9.6e-2 & 0.16 & 2.43 & 6.83 & 8.07 \\
        \cline{2-8}
        & \sysname & 1.3e-2 & 9.6e-2 & 0.17 & 2.49 & 7.64 & 8.07 \\
        \hline
        \multirow{2}{*}{batched} & ABY3 & 1.45 & 12.6 & 13.2 & 8.12 & 58.9 & 59.5 \\
        \cline{2-8}
        & \sysname & 1.38 & 12.02 & 12.2 & 7.22 & 56.3 & 57.9 \\
        \hline
	\end{tabular}
	\caption{Time (seconds) for inference of real algorithms with different engines.}
	\label{table:time_inference}
\end{table}

Finally, we stress the usability of our front-end. Table~\ref{table:lines_of_codes} shows the lines of codes for each algorithm and the time for a student who focuses on data mining but is unfamiliar with cryptography to write each algorithm using our front-end.

\begin{table}[tb]
	\centering
    \small
    \begin{tabular}{|c|c|c|c|c|c|}
		\hline
         & \tabincell{c}{LR} & \tabincell{c}{MF} & NN & CNN & \tabincell{c}{CNN+BN} \\
		\hline
        lines & 42 & 25 & 9 & 83 & 87 \\
		\hline
        time & 0.7 & 0.5 & 0.1 & 1.5 & 1.5 \\
        \hline
	\end{tabular}
	\caption{Lines of codes and time (hours) for implementing real algorithms.}
	\label{table:lines_of_codes}
\end{table}

\section{Conclusion and Future Work}
\label{section:conclusion_future}

Over thirty years of MPC literature provides an ocean of protocols and systems great on certain aspects of performance, security or ease of programming.  
We believe it is time to integrate these techniques into an application-driven and coherent  system for machine learning tasks.  
\sysname is a framework with top-down design.  
At the top, it provides familiar Python-compatible interfaces with essential data types like real numbers and arrays, and use  code optimizer/checkers to avoid common mistakes.  
In the middle, using an intermediary for storage and communication, we build a composable PO system that helps decoupling the front-end with backend.  
At the low level, we design new protocols that improve computation speed.  
\sysname shows great potential: it handles large data set (1M-by-5K) and complex algorithms (CNN) fast, with minimal program porting effort.

\sysname opens up many future directions. 
Firstly, we are improving the \sysname computation engine to provide active security while preserving high efficiency. 
Secondly, we would like to port existing machine learning libraries to our front-end. 
Thirdly, we will support more computation engines.
Fourthly, although we focus on MPC in this work, we will introduce randomization to protect the final results~\cite{MPC+DP, PEM}. 
Last but not least, we will also improve fault tolerance mechanism to the servers.

\newpage
\newpage
\appendix

\begin{table}[htbp]
	\begin{tabular}{|c|l|}
		\hline
		$\phi$ & the big prime that determins the field \\
		\hline
		$\mathbb{Z}_\phi$ & the additive group of integers module $\phi$ \\
		\hline
		$b$ & the bound of numbers in the computation \\
		\hline
		$k$ & the scaling factor \\
		\hline
		$S(\cdot)$ & \tabincell{l}{the secret sharing function splitting integers \\ in $\mathbb{Z}_\phi$ to shares} \\
		\hline
        $\map{x}$ & \tabincell{l}{the corresponding integer in $\mathbb{Z}_\phi$ of \\ a real number $x$} \\
		\hline
        $\share{x}$ & \tabincell{l}{the secret sharing result of a real number $x$ \\ that is equivalent to $S(\map{x})$} \\
		\hline
        $\share{\cdot, \cdot}^{-1}$ & \tabincell{l}{the reverse process of $\share{\cdot}$ that maps \\ the secrets back to real numbers} \\
		\hline
		$I(\cdot)$ & \tabincell{l}{the helper function that converts integers in \\ $\mathbb{Z}_\phi$ to the signed representation}  \\
		\hline
		$x, y$ & \tabincell{l}{private variables} \\
		\hline
		$\tabincell{c}{$x_1, x_2$ \\ ${y_1, y_2}$}$ & \tabincell{l}{the shares of private variables} \\
		\hline
	\end{tabular}
	\caption{The notations in this paper.}
	\label{table:notation}
\end{table}

\subsection{Proof of security of the fixed-point multiplication protocol}
\label{appendix:proof_security_multiplication}

To argue the security of the fixed-point multiplication protocol, we define functionality $\mathcal{F}_{mult}$ in the ideal model, and prove that it is indistinguishable from the real protocol by constructing efficient simulators.

\noindent \begin{minipage}[h]{0.47\textwidth}
\fbox{
\parbox{\textwidth}{
\begin{center}
Functionality $\mathcal{F}_{mult}$
\end{center}
    After receiving $(x_1, x_1'), (y_1, y_1')$ from \server{1}, $(x_2, x_2'), (y_2, y_2')$ from \server{2}, $(x_a, x_a'), (y_a, y_a')$ from \server{a} and $(x_b, x_b'), (y_b, y_b')$ from \server{b}, 
    it does the following: \\
    \begin{enumerate}
        \item Set $x = x_1 + x_2$ and $y = y_1 + y_2$. \\
        \item Set $z = xy / 2^d$. \\
        \item Sample $z_1 \in \mathbb{Z}_{\phi}$ and $z_1' \in \mathbb{Z}_{\phi}$, and set $z_2 = z - z_1$ and $z_2' = z - z_1'$. \\
        \item Set $z_a = z_2, z_a' = z_1', z_b = z_1, z_b' = z_2'$. \\
        \item Send $(z_1, z_1')$, $(z_2, z_2')$, $(z_a, z_a')$ and $(z_b, z_b')$ to \server{1}, \server{2}, \server{a} and \server{b}, respectively.
    \end{enumerate}
}
}
\vspace{+0.1in}
\end{minipage}

The view of the \server{i} party ($i\in\{1,2,a,b\}$) during an execution of a protocol $\pi$ is denoted by $View^{\pi}$ and the output is denoted by $Output^{\pi}_i$.
We often ommit the superscript $\pi$ for simplicity.

\begin{defi} \label{def:semi-honest}
    Let $\pi$ be a protocol. We say that \textbf{$\pi$} securely realizes $\mathcal{F}_{mult}$ in the presence of semi-honest adversaries, 
    if for each $i \in \{1, 2, a, b\}$ there exists a probabilistic polynomial-time algorithm $Sim_i$ such that the real view and the corresponding simulated view are computationally indistinguishable.
\end{defi}

\begin{theo} \label{theo:Protocol1_security}
  Protocol 1 securely realizes $\mathcal{F}_{mult}$ presence of semi-honest adversaries.
\end{theo}

\para{Proof.}
Our proof is similar to the proof in \cite{IMPROVED_3PC}.
Specifically, we first construct an efficient simulator for \server{1}, which are referred to as $Sim_1$ and receives the input $X_1 = (x_1, x_1')$ and $Y_1 = (y_1, y_1')$, as well as the seed for generating pseudo-random numbers. 
$Sim_1$ generates $r^*_{12}$ and ${r'}^*_{12}$ using the seed, then calculates $t^*_1 = x_1y_1' - r^*_{12}$ and ${t'}^*_1 = x_1'y_1 - {r'}^*_{12}$. 
$Sim_1$ samples random numbers ${t'}^*_a \in \mathbb{Z}_\phi$ and $t^*_b \in \mathbb{Z}_\phi$, and sets $z^*_1 = (t^*_b + t^*_1) / 2^d$ and ${z'}^*_1 = ({t'}^*_a + {t'}^*_1) / 2^d$, then sends $z^*_1$ and ${z'}^*_1$ to $\mathcal{F}_{mult}$. 
$Sim_1$ adds $r^*_{12}$ and ${r'}^*_{12}$, ${t'}^*_a$ and $t^*_b$ to the view of \server{1}.
In a real execution, $r_{12}$ and ${r'}_{12}$ are pseudo-random numbers that are generated using the same seed from \server{1}, thus are equal to $r^*_{12}$ and ${r'}^*_{12}$. 
On the other hand, $t'_a$ and $t_b$ are masked by pseudo-random numbers generated by $S_a$ and $S_b$. 
As the pseudo-random numbers are generated using seeds unknown by \server{1}, $t'_a$ and $t_b$ are indistinguishable with truely-random numbers $({t'}^*_a$ and $t^*_b)$ for \server{1}. 
We can then conclude that $View_1$ and $Sim_1$ are indistinguishable.
The simulators for the other servers can be constructed in the same way.
$\qed$

\subsection{Numpy features implemented in \sysname}
\label{appendix:ndarray}
\alert{In \sysname front-end, we provide two Numpy features widely utilized to implement machine learning algorithms: \emph{broadcasting} and \emph{ndarray} methods. 

Broadcasting allows operations between arrays of different shapes, by ``broadcasting'' the smaller one automatically, as long as their dimensionalities match (see \cite{BROADCASTING} for details). For example, given a scalar $x$, a $4 \times 3$ array $A$, a $2 \times 4 \times 3$ array $B$ and a $2 \times 1 \times 3$ array $C$, the expressions $x \bigodot A$, $A \bigodot B$ and $B \bigodot C$ are all legal in \sysname, where $\bigodot$ can be $+, \times \text{and} >$ etc. Note that in \sysname, the above variables can be either public or private.

We also implement most of the ndarray methods of Numpy, with which application programmers can manipulate arrays conveniently and efficiently, except for the methods related with IO (we leave IO as the future work). Table~\ref{table:ndarray} lists the ndarray methods we have implemented. Please see \cite{NDARRAY} for details of \code{numpy.ndarray}.}

\begin{table}[htbp]
\footnotesize
\begin{tabular}
{|c|c|c|c|}
\hline
\code{all} & \code{any} & \code{append} & \code{argmax} \\
\hline
\code{argmin} & \code{argparition} & \code{argsort} & \code{clip} \\
\hline
\code{compress} & \code{copy} & \code{cumprod} & \code{cumsum} \\
\hline
\code{diag} & \code{dot} & \code{fill} & \code{flatten} \\
\hline
\code{item} & \code{itemset} & \code{max} & \code{mean} \\
\hline
\code{min} & \code{ones} & \code{outer} & \code{partition} \\
\hline
\code{prod} & \code{ptp} & \code{put} & \code{ravel} \\
\hline
\code{repeat} & \code{reshape} & \code{resize} & \code{searchsorted} \\
\hline
\code{sort} & \code{squeeze} & \code{std} & \code{sum} \\
\hline
\code{swapaxes} & \code{take} & \code{tile} & \code{trace} \\
\hline
\code{transpose} & \code{var} & \code{zeros} & \code{} \\
\hline
\end{tabular}
\caption{The ndarray methods implemented in \sysname.}
\label{table:ndarray}
\end{table}

\newpage
{\footnotesize \bibliographystyle{acm}
\bibliography{reference.bib}}

\begin{thebibliography}{10}

\bibitem{TENSORFLOW}
{\sc Abadi, M., Barham, P., Chen, J., Chen, Z., Davis, A., Dean, J., Devin, M.,
  Ghemawat, S., Irving, G., and Isard, M.}
\newblock Tensorflow: a system for large-scale machine learning.

\bibitem{GENERALIZING_SPDZ}
{\sc Araki, T., Barak, A., Furukawa, J., Keller, M., Lindell, Y., Ohara, K.,
  and Tsuchida, H.}
\newblock Generalizing the spdz compiler for other protocols.
\newblock In {\em Proceedings of the 2018 ACM SIGSAC Conference on Computer and
  Communications Security\/} (2018), ACM, pp.~880--895.

\bibitem{IMPROVED_3PC}
{\sc Araki, T., Furukawa, J., Lindell, Y., Nof, A., and Ohara, K.}
\newblock High-throughput semi-honest secure three-party computation with an
  honest majority.
\newblock In {\em ACM Sigsac Conference on Computer and Communications
  Security\/} (2016), pp.~805--817.

\bibitem{FairplayMP}
{\sc Ben-David, A., Nisan, N., and Pinkas, B.}
\newblock {FairplayMP: a system for secure multi-party computation}.
\newblock In {\em CCS '08\/} (2008), ACM.

\bibitem{MATRIX_FACTORIZATION_DP}
{\sc Berlioz, A., Friedman, A., Kaafar, M.~A., Boreli, R., and Berkovsky, S.}
\newblock Applying differential privacy to matrix factorization.
\newblock In {\em The ACM Conference\/} (2015), pp.~107--114.

\bibitem{SECUREC}
{\sc Bogdanov, D., Laud, P., and Randmets, J.}
\newblock Domain-polymorphic programming of privacy-preserving applications.
\newblock In {\em Proceedings of the Ninth Workshop on Programming Languages
  and Analysis for Security\/} (2014), ACM, p.~53.

\bibitem{SHAREMIND}
{\sc Bogdanov, D., Laur, S., and Willemson, J.}
\newblock Sharemind: A framework for fast privacy-preserving computations.
\newblock In {\em European Symposium on Research in Computer Security\/}
  (2008), Springer, pp.~192--206.

\bibitem{ENCRYPTED_CLASSIFICATION}
{\sc Bost, R., Popa, R.~A., Tu, S., and Goldwasser, S.}
\newblock Machine learning classification over encrypted data.
\newblock In {\em NDSS\/} (2015).

\bibitem{IMPROVED_FIXED}
{\sc Catrina, O., and De~Hoogh, S.}
\newblock Improved primitives for secure multiparty integer computation.
\newblock In {\em International Conference on Security and Cryptography for
  Networks\/} (2010), Springer, pp.~182--199.

\bibitem{FIXED_POINT}
{\sc Catrina, O., and Saxena, A.}
\newblock Secure computation with fixed-point numbers.
\newblock In {\em International Conference on Financial Cryptography and Data
  Security\/} (2010), Springer, pp.~35--50.

\bibitem{HANDWRITTEN}
{\sc Chen, Z.}
\newblock Handwritten digits recognition.
\newblock In {\em International Conference on Image Processing, Computer
  Vision, \& Pattern Recognition, Ipcv 2009, July 13-16, 2009, Las Vegas,
  Nevada, Usa, 2 Volumes\/} (2000), pp.~690--694.

\bibitem{SPDZ_COMPILER}
{\sc Damg{\aa}rd, I., Keller, M., Larraia, E., Pastro, V., Scholl, P., and
  Smart, N.~P.}
\newblock Practical covertly secure mpc for dishonest majority--or: breaking
  the spdz limits.
\newblock In {\em European Symposium on Research in Computer Security\/}
  (2013), Springer, pp.~1--18.

\bibitem{SPDZ}
{\sc Damg{\aa}rd, I., Pastro, V., Smart, N., and Zakarias, S.}
\newblock Multiparty computation from somewhat homomorphic encryption.
\newblock In {\em Advances in Cryptology--CRYPTO 2012}. Springer, 2012,
  pp.~643--662.

\bibitem{AUTOMATED_SYNTHESIS}
{\sc Demmler, D., Dessouky, G., Koushanfar, F., Sadeghi, A.-R., Schneider, T.,
  and Zeitouni, S.}
\newblock Automated synthesis of optimized circuits for secure computation.
\newblock In {\em Proceedings of the 22nd ACM SIGSAC Conference on Computer and
  Communications Security\/} (2015), ACM, pp.~1504--1517.

\bibitem{ABY}
{\sc Demmler, D., Schneider, T., and Zohner, M.}
\newblock Aby-a framework for efficient mixed-protocol secure two-party
  computation.
\newblock In {\em NDSS\/} (2015).

\bibitem{P4P}
{\sc {Duan, Yitao and Canny, John and Zhan, Justin}}.
\newblock {P4P: Practical Large-scale Privacy-preserving Distributed
  Computation Robust Against Malicious Users}.
\newblock In {\em Proceedings of the 19th USENIX Conference on Security\/}
  (2010), USENIX Security'10, USENIX Association.

\bibitem{HELIB}
{\sc Halevi, S., and Shoup, V.}
\newblock Algorithms in helib.
\newblock In {\em International Cryptology Conference\/} (2014), Springer,
  pp.~554--571.

\bibitem{MOVIELENS}
{\sc Harper, F.~M., and Konstan, J.~A.}
\newblock The movielens datasets: History and context.
\newblock {\em ACM Transactions on Interactive Intelligent Systems (TiiS) 5}, 4
  (2016), 19.

\bibitem{TASTY}
{\sc Henecka, W., Sadeghi, A.-R., Schneider, T., Wehrenberg, I., et~al.}
\newblock Tasty: tool for automating secure two-party computations.
\newblock In {\em Proceedings of the 17th ACM conference on Computer and
  communications security\/} (2010), ACM, pp.~451--462.

\bibitem{BATCHNORMALIZATION}
{\sc Ioffe, S., and Szegedy, C.}
\newblock Batch normalization: Accelerating deep network training by reducing
  internal covariate shift.
\newblock In {\em International Conference on Machine Learning\/} (2015),
  pp.~448--456.

\bibitem{BROADCASTING}
{\sc Jones, E., Oliphant, T., Peterson, P., et~al.}
\newblock numpy.ndarray.
\newblock \url{https://docs.scipy.org/doc/numpy/user/basics.broadcasting.html},
  2011--.

\bibitem{NDARRAY}
{\sc Jones, E., Oliphant, T., Peterson, P., et~al.}
\newblock numpy.ndarray.
\newblock
  \url{https://docs.scipy.org/doc/numpy/reference/generated/numpy.ndarray.html},
  2011--.

\bibitem{SECURE_FLOATING}
{\sc Kamm, L., and Willemson, J.}
\newblock Secure floating point arithmetic and private satellite collision
  analysis.
\newblock {\em International Journal of Information Security 14}, 6 (2015),
  531--548.

\bibitem{IMPROVED_GC}
{\sc Kolesnikov, V., Sadeghi, A.-R., and Schneider, T.}
\newblock Improved garbled circuit building blocks and applications to auctions
  and computing minima.
\newblock In {\em International Conference on Cryptology and Network
  Security\/} (2009), Springer, pp.~1--20.

\bibitem{KSS}
{\sc Kreuter, B., Shelat, A., and Shen, C.-H.}
\newblock Billion-gate secure computation with malicious adversaries.
\newblock In {\em USENIX Security Symposium\/} (2012), vol.~12, pp.~285--300.

\bibitem{HYBRID}
{\sc Krips, T., and Willemson, J.}
\newblock Hybrid model of fixed and floating point numbers in secure multiparty
  computations.
\newblock In {\em International Conference on Information Security\/} (2014),
  Springer, pp.~179--197.

\bibitem{MNIST}
{\sc Lecun, Y., and Cortes, C.}
\newblock The mnist database of handwritten digits.
\newblock \url{http://yann.lecun.com/exdb/mnist}, 2010.

\bibitem{LENET5}
{\sc LeCun, Y., et~al.}
\newblock Lenet-5, convolutional neural networks.
\newblock {\em URL: http://yann. lecun. com/exdb/lenet\/} (2015).

\bibitem{PEM}
{\sc Li, Y., Duan, Y., and Xu, W.}
\newblock Pem: Practical differentially private system for large-scale
  cross-institutional data mining.
\newblock In {\em Joint European Conference on Machine Learning and Knowledge
  Discovery in Databases\/} (2017), Springer.

\bibitem{OBLIVM}
{\sc Liu, C., Wang, X.~S., Nayak, K., Huang, Y., and Shi, E.}
\newblock Oblivm: A programming framework for secure computation.
\newblock In {\em Security and Privacy (SP), 2015 IEEE Symposium on\/} (2015),
  IEEE, pp.~359--376.

\bibitem{OBLIVIOUS_NN}
{\sc Liu, J., Juuti, M., Lu, Y., and Asokan, N.}
\newblock Oblivious neural network predictions via minionn transformations.
\newblock In {\em Proceedings of the 2017 ACM SIGSAC Conference on Computer and
  Communications Security\/} (2017), ACM, pp.~619--631.

\bibitem{FHE_STATISTICS}
{\sc Lu, W., Kawasaki, S., and Sakuma, J.}
\newblock Using fully homomorphic encryption for statistical analysis of
  categorical, ordinal and numerical data.
\newblock {\em IACR Cryptology ePrint Archive 2016\/} (2016), 1163.

\bibitem{ABY3}
{\sc Mohassel, P., and Rindal, P.}
\newblock Aby 3: a mixed protocol framework for machine learning.
\newblock In {\em Proceedings of the 2018 ACM SIGSAC Conference on Computer and
  Communications Security\/} (2018), ACM, pp.~35--52.

\bibitem{GC_3PC}
{\sc Mohassel, P., Rosulek, M., and Zhang, Y.}
\newblock Fast and secure three-party computation:the garbled circuit approach.
\newblock In {\em The ACM Sigsac Conference\/} (2015), pp.~591--602.

\bibitem{SECUREML}
{\sc Mohassel, P., and Zhang, Y.}
\newblock Secureml: A system for scalable privacy-preserving machine learning.
\newblock In {\em 2017 IEEE Symposium on Security and Privacy (SP)\/} (May
  2017), pp.~19--38.

\bibitem{AST}
{\sc Neamtiu, I., Foster, J.~S., and Hicks, M.}
\newblock Understanding source code evolution using abstract syntax tree
  matching.
\newblock In {\em International Workshop on Mining Software Repositories, MSR
  2005, Saint Louis, Missouri, Usa, May\/} (2005), pp.~1--5.

\bibitem{PYTORCH}
{\sc Paszke, A., Gross, S., Chintala, S., Chanan, G., Yang, E., DeVito, Z.,
  Lin, Z., Desmaison, A., Antiga, L., and Lerer, A.}
\newblock Automatic differentiation in pytorch.

\bibitem{SCIKIT_LEARN}
{\sc Pedregosa, F., Varoquaux, G., Gramfort, A., Michel, V., Thirion, B.,
  Grisel, O., Blondel, M., Prettenhofer, P., Weiss, R., Dubourg, V.,
  Vanderplas, J., Passos, A., Cournapeau, D., Brucher, M., Perrot, M., and
  Duchesnay, E.}
\newblock Scikit-learn: Machine learning in {P}ython.
\newblock {\em Journal of Machine Learning Research 12\/} (2011), 2825--2830.

\bibitem{MPC+DP}
{\sc Pettai, M., and Laud, P.}
\newblock Combining differential privacy and secure multiparty computation.
\newblock In {\em Proceedings of the 31st Annual Computer Security Applications
  Conference\/} (2015), ACM, pp.~421--430.

\bibitem{OT}
{\sc Rabin, M.~O.}
\newblock How to exchange secrets by oblivious transfer.
\newblock Tech. Rep. TR-81, Aiken Computation Laboratory, Harvard University,
  1981.

\bibitem{L1}
{\sc Schropfer, A., Kerschbaum, F., and Muller, G.}
\newblock L1-an intermediate language for mixed-protocol secure computation.
\newblock In {\em Computer Software and Applications Conference, IEEE 35th
  Annual\/} (2011), pp.~298--307.

\bibitem{SHAMIR}
{\sc Shamir, A.}
\newblock {How to share a secret}.
\newblock {\em Communications of the ACM 22}, 11 (1979).

\bibitem{EULER_METHOD}
{\sc Stoer, J., and Bulirsch, R.}
\newblock Introduction to numerical analysis.
\newblock {\em Mathematics of Computation 24}, 111 (1980), 749.

\bibitem{NUMPY}
{\sc Walt, S. V.~D., Colbert, S.~C., and Varoquaux, G.}
\newblock The numpy array: A structure for efficient numerical computation.
\newblock {\em Computing in Science \& Engineering 13}, 2 (2011), 22--30.

\bibitem{EMP}
{\sc Wang, X.}
\newblock emp-sh2pc.
\newblock \url{https://github.com/emp-toolkit/emp-sh2pc}, 2016.

\bibitem{VECTORIZATION}
{\sc Weinhardt, M., and Luk, W.}
\newblock Pipeline vectorization.
\newblock {\em Computer-Aided Design of Integrated Circuits and Systems, IEEE
  Transactions on 20}, 2 (2001), 234--248.

\bibitem{PPDM_STATUS}
{\sc Wu, X., Chu, C.~H., Wang, Y., Liu, F., and Yue, D.}
\newblock Privacy preserving data mining research: Current status and key
  issues.
\newblock {\em Lecture Notes in Computer Science 4489\/} (2007), 762--772.

\bibitem{YAO_PROTOCOL}
{\sc Yao, A.~C.}
\newblock Protocols for secure computations.
\newblock {\em Foundations of Computer Science Annual Symposium on\/} (1982),
  160--164.

\bibitem{NEWTON_RAPSHON}
{\sc Ypma, T.~J.}
\newblock Historical development of the newton-raphson method.
\newblock {\em Siam Review 37}, 4 (1995), 531--551.

\bibitem{OBLIV_C}
{\sc Zahur, S., and Evans, D.}
\newblock Obliv-c: A language for extensible data-oblivious computation.
\newblock {\em IACR Cryptology ePrint Archive 2015\/} (2015), 1153.

\bibitem{PICCO}
{\sc Zhang, Y., Steele, A., and Blanton, M.}
\newblock Picco: a general-purpose compiler for private distributed
  computation.
\newblock In {\em Proceedings of the 2013 ACM SIGSAC conference on Computer \&
  communications security\/} (2013), ACM, pp.~813--826.

\end{thebibliography}

\end{document}